\documentclass[12pt]{article}

\usepackage[latin9]{inputenc}
\usepackage{rotating}

\usepackage{amsmath,amsfonts,bm,amssymb}
\usepackage{amsthm}

\usepackage[a4paper,left=1.2in,right=1.2in,top=1.25in,bottom=1.25in]{geometry}

\usepackage{setspace}
\usepackage{graphicx}
\usepackage{subfigure}
\graphicspath{{Graphs/},{../Graphs/}}

\usepackage{setspace}
\usepackage{wrapfig}
\usepackage[shortlabels]{enumitem}

\usepackage{float}

\usepackage{units}

\usepackage[multiple]{footmisc}

\usepackage{booktabs}
\usepackage{tabularx}
\usepackage{multirow}
\usepackage{array}
\newcolumntype{L}[1]{>{\raggedright\let\newline\\\arraybackslash\hspace{0pt}}m{#1}}
\newcolumntype{C}[1]{>{\centering\let\newline\\\arraybackslash\hspace{0pt}}m{#1}}
\newcolumntype{R}[1]{>{\raggedleft\let\newline\\\arraybackslash\hspace{0pt}}m{#1}}

\usepackage{color}
\usepackage[pdftex]{hyperref}
\usepackage[capitalise]{cleveref}

\usepackage{natbib}
\bibpunct{(}{)}{;}{a}{,}{,}

\usepackage{authblk}

\newtheorem{proposition}{{\bf \sc Proposition}}

\theoremstyle{remark}

\newcommand{\bge}{\begin{equation}}
\newcommand{\ene}{\end{equation}}

\usepackage[final]{changes}

\begin{document}

    \title{Network Effects of Tariffs\thanks{%
    I thank Alberto Dalmazzo and Roberto Rozzi for very helpful comments.}}
	
	\author[a,b,*]{Paolo Pin} 
    	\affil[a]{Dipartimento di Economia Politica e Statistica, Universit\`a di Siena}
        \affil[b]{BIDSA, Universit\`a Bocconi}
	\affil[*]{Corresponding author. Email: paolo.pin@unisi.it}	
	
    \date{April 2025}
\maketitle
	
\vspace*{-\baselineskip}
			
\begin{abstract}
We develop a model in which country-specific tariffs shape trade flows, prices, and welfare in a global economy with one homogeneous good. Trade flows form a Directed Acyclic Graph (DAG), and tariffs influence not only market outcomes but also the structure of the global trade network. A numerical example illustrates how tariffs may eliminate targeted imports, divert trade flows toward third markets, expose domestic firms to intensified foreign competition abroad, reduce consumer welfare, and ultimately harm the country imposing the tariff.

\end{abstract}

\small \noindent \textbf{Keywords:} Tariffs, Trade Networks, Directed Acyclic Graphs.

\small \noindent \textbf{JEL Classification:} F13 (Trade Policy; International Trade Organizations), D85 (Networks and Network Formation).

\section{Introduction}

The role of tariffs in shaping global trade patterns has attracted renewed attention, following episodes of rising protectionism, trade wars, and the use of tariffs as tools of political pressure. A series of empirical and policy discussions has highlighted how tariffs affect not only bilateral trade volumes but also broader macroeconomic variables such as aggregate demand, inflation, and sectoral production structures. These contributions emphasize key mechanisms such as asymmetric tariff pass--through, the disruption of global value chains, and the potential for retaliatory measures to further destabilize trade relations.

However, much of this discussion remains partial and linear: tariffs are often analyzed in isolation or as shocks to fixed structures. What is often missing is a general equilibrium perspective that takes into account how tariffs endogenously alter the network of global trade flows itself. When tariff changes induce countries to redirect their exports, the resulting reconfiguration of trade networks can lead to discontinuous, non--linear effects on prices, quantities, and welfare across the system.

Recent work has addressed these issues. \citet{fajgelbaum2019return} study the distributional effects of tariffs in general equilibrium. \citet{amiti2019effects} provide empirical evidence on asymmetric tariff pass--through between China and the United States.
\cite{grossman2024tariffs} study the effects of unanticipated tariffs.
\citet{blanchard2025global} highlight how supply chain disruptions propagate the effects of tariffs across sectors. Earlier foundational work by \citet{grossman1994protection} models endogenous tariff-setting in a political economy framework. Although these contributions provide important insights, a fully network-based general equilibrium analysis, where trade flows endogenously respond to strategic tariff choices, has received comparatively less attention.

This short paper takes a small step in that direction. We develop a simple model in which tariffs influence not only prices and welfare, but also the structure of global trade networks. Trade flows form a Directed Acyclic Graph (DAG), and tariff changes can induce discrete reconfigurations of this network. A fictitious numerical example, based on  trade between the EU, USA, and China, illustrates how  tariff changes can eliminate trade links, alter equilibrium prices and quantities, and generate unintended welfare effects -- including intensified foreign competition for domestic firms.

Understanding how tariffs reshape the global trade network, and how these changes feed back into welfare outcomes, is essential for a comprehensive evaluation of trade policy. Our framework provides a foundation for future work on endogenous tariff-setting as a network formation game, where countries strategically influence not just bilateral flows but the entire structure of global trade.
Because there is limited historical evidence on sudden, large--scale tariff changes and experimental validation is not feasible, theoretical modeling becomes essential. For reasons of tractability and interpretability, in the next section we focus on a simple model that is analytically solvable. Section \ref{example} provides an example and Section \ref{conclusion} concludes.

\section{A Model of Tariffs}
\label{model}

We develop a simple model in which country-specific tariffs shape trade flows, prices, and welfare in a global economy with $n$ countries and one homogeneous good.

Each country $i$ is characterized by an indirect supply function $s_i (q): \mathbb{R}_+ \rightarrow \mathbb{R}_+$ and an indirect demand function $d_i (q): \mathbb{R}_+ \rightarrow \mathbb{R}_+$. 
$s_i (q)$ is continuous and strictly increasing.
$d_i (q)$ is continuous and strictly decreasing for every $q$ such that $d_i (q)>0$, and $\lim_{q \rightarrow \infty} d_i (q)=0$.

Let \( t_{ij} \) denote the ad valorem (i.e., value-based) tariff imposed by country \( i \) on imports from country \( j \). Collect all such tariffs into a matrix \( \mathbf{t} \in \mathbb{R}_+^{n \times n} \). This specification is general: off-diagonal entries \( t_{ij} \) with \( i \neq j \) represent bilateral tariffs on international trade, while we set diagonal terms \( t_{ii} =0\) for every country $i$. 

Let \( \mathbf{q} \in \mathbb{R}_+^{n \times n} \) denote the matrix of trade flows, where \( q_{ij} \) represents the quantity produced in country \( j \) and consumed in country \( i \). Define:
\begin{itemize}
  \item \( \mathbf{p}_c^T = (p^T_{1,c}, \dots, p^T_{n,c})^\top \): vector of consumer prices,
  \item \( \mathbf{p}_f^T = (p^T_{1,f}, \dots, p^T_{n,f})^\top \): vector of producer prices.
\end{itemize}

We define here below a \emph{tariff equilibrium}.
Equilibrium is determined by two conditions: consumer market clearing and firm destination choice. Consumers demand goods according to local prices, while producers allocate supply toward the most profitable destinations after adjusting for tariffs.
In each country \( i \), the \textbf{consumer market clearing condition} requires that total demand matches the quantity available for consumption. Given the indirect demand function, this implies:
\begin{equation}
    \label{clearing}
p^T_{i,c} = d_i \left( \sum_j q_{ij} \right) ,
\end{equation}
where the sum is over all sources of supply (both domestic and foreign).

Similarly, each country \( j \) must allocate its production across all destination markets. The \textbf{firm selection condition} states that supply is determined by:
\begin{equation}
    \label{selling}
p^T_{j,f} =  s_j \left( \sum_i q_{ij} \right) ,
\end{equation}
subject to firms in country \( j \) choosing to serve only those destinations \( i \) that yield the highest \emph{effective revenue} per unit, defined as the consumer price adjusted for the tariff $
\frac{p^T_{i,c}}{1+ t_{ij}}$.
Formally, this selection condition can be written as:
\begin{equation}
    \label{selection}
p^T_{j,f} = \arg\max_h \left\{ \frac{p^T_{h,c}}{1 + t_{hj}} \right\}
\end{equation}

That is, producers in country \( j \) supply only the markets -- whether domestic or foreign -- where their net revenue is maximal. If multiple destinations yield the same maximum effective revenue, firms are indifferent among them; otherwise, trade flows are directed exclusively toward the most profitable markets.

These two conditions jointly determine a tariff equilibrium in which each representative firm supplies only the most profitable destination markets, taking prices as given, and consumer prices adjust to clear domestic markets given the available supplies. Firms allocate their exports toward destinations offering the highest effective revenues (consumer prices adjusted for tariffs). The resulting allocation of trade flows shapes the global trade network and determines market outcomes.



We now characterize the existence and structure of a tariff equilibrium in terms of trade flows.
Formally, the global trade structure can be represented by a \emph{directed network}, where nodes correspond to countries, and a directed link from node \( j \) to node \( i \) is present if and only if \( q_{ij} > 0 \).

We need another definition: a \textbf{Directed Acyclic Graph (DAG)} is a directed network that contains no directed cycles. That is, it is impossible to start at a node, follow a sequence of directed links, and return to the starting node.

\begin{proposition}
\label{prop:dag_existence}
Given any set of primitives of the model \( \left( \left\{ d_i \right\}_{i \in \{1 ,\dots, n  \}} , \left\{ s_i \right\}_{i \in \{1 ,\dots, n  \}}  , \mathbf{t} \right) \), there exists a tariff equilibrium.  
Moreover, if \( t_{ij} > 0 \) for any pair of countries \( i \ne j \), the equilibrium  network is a DAG.
\end{proposition}

\begin{proof}
See Appendix~\ref{proof:prop_dag_existence}.
\end{proof}

The firm selection condition imposes a strong constraint on trade flows: it prevents the existence of trade cycles. If a country exports to another, and that country exports to a third, and so on, trade cannot eventually flow back to the original exporter. Intuitively, firms always choose the most profitable market, and exporting along a cycle would contradict profit maximization. Since this reasoning applies to cycles of any length, for any pair of countries \(i\) and \(j\), it cannot be that both \(q_{ij}\) and \(q_{ji}\) are strictly positive.\footnote{A similar argument appears in related contexts, such as the formation of interbank credit networks under bilateral profit maximization (see \citealp{AnufrievDeghiPanchenkoPin2025}).}

DAG structures are also fragile to perturbations \citep{presanis2013conflict,aleta2022cycle}: small changes in tariffs or trade costs can modify effective revenues, introduce cycles, and violate acyclicity. Restoring a DAG may require extensive reconfiguration of trade flows, so the effects of local policy changes can propagate unpredictably across the global trade network.

The equilibrium condition \eqref{selection} applies to any exporting country, whether a pure exporter or a transit country. Producer prices are endogenously determined by the most profitable export opportunity, so exporters behave as price takers relative to international market conditions. This has two main implications. First, the trade network topology compresses domestic margins, as producer prices are anchored to the best external opportunity net of tariffs. Second, exporters are highly sensitive to external shocks: changes in tariffs or consumer prices elsewhere in the network immediately affect domestic producer prices.
It can even happen that, if foreign destinations offer higher effective revenues than domestic sales, firms will exclusively serve external markets, and domestic consumption will rely entirely on imports.

\subsection{Welfare Effects of Tariffs}
\label{sec:welfare_effects}

Having characterized the structure of trade flows, we now turn to welfare analysis.  
In the presence of trade, welfare in each country has three components:
\begin{itemize}
    \item \textbf{Consumer surplus}: the net benefit that consumers derive from purchasing the good at the prevailing consumer price.
    \item \textbf{Firm profits}: the net benefit that producers obtain from selling the good at the prevailing producer price.
    \item \textbf{Government revenues}: income collected through tariffs on imports.
\end{itemize}

Formally, the consumer surplus in country \(i\) is given by
\[
w^T_{i,c} = \int_0^{\sum_j q_{ij}} \left( d_i (x) - p_{i,c}^T \right) dx = 
\int_0^{d^{-1}_i \left( p_{i,c}^T \right) } \left( d_i (x) - p_{i,c}^T \right) dx , 
\]
while firm profits are
\[
w^T_{i,f} = \int_0^{\sum_j q_{ji}} \left(p_{i,f}^T - s_i (x) \right) dx
= \int_0^{s^{-1}_i \left( p_{i,f}^T \right) } \left(p_{i,f}^T - s_i (x) \right) dx
.
\]
Government revenues collected by country \(i\) are
\[
r^T_i = \sum_j t_{ij} \, p^T_{j,f} \, q_{ij}.
\]
Hence, total welfare in country \(i\) is
\[
W^T_i = w^T_{i,c} + w^T_{i,f} + r^T_i.
\]

\medskip

We now investigate how changes in tariffs affect welfare.  
Consider a marginal increase in \( t_{ij} \), the tariff imposed by country \( i \) on imports from country \( j \).  As long as the structure of the trade network remains unchanged---that is, the set of active trade links is unaffected---the welfare effects are relatively predictable. In country \(i\), firms benefit from reduced foreign competition, whereas consumers are harmed by higher import prices. In country \(j\), firms suffer from the loss of market access, while consumers may gain if exporters lower their prices in response to diminished external demand. The effect on government revenues in both countries remains ambiguous, depending on the magnitude of trade volume adjustments and resulting price changes.

However, if the tariff increase is large enough to change the structure of the trade network---for instance, by making the link from \(j\) to \(i\) unprofitable---the resulting reconfiguration of trade flows can lead to non-marginal, and possibly unintended, consequences.  
Such discrete changes in the network may produce discontinuities in equilibrium prices, quantities, and welfare levels, making comparative statics highly non-linear and policy outcomes harder to anticipate.

\medskip

The following example illustrates these mechanisms with a simple three-country model.

\section{Example}
\label{example}

To illustrate the mechanisms outlined above, we present a fully fictitious numerical example of trade flows and welfare in a three-country economy.

We consider a hypothetical homogeneous good -- wine -- traded among three regions: the European Union (EU, Country 1), the United States (USA, Country 2), and China (Country 3). Imagine that prices are expressed in US dollars per bottle, and quantities are measured in millions of bottles. The example is entirely illustrative.

The indirect supply and demand functions for each region are linear:

\begin{table}[h]
\centering
\begin{tabular}{@{}lll@{}}
\toprule
\textbf{Country} & \textbf{Supply Function} & \textbf{Demand Function} \\
\midrule
EU (1) & $s_1(q) = 2 + 0.5q$ & $d_1(q) = 8 - q$ \\
USA (2) & $s_2(q) = 3 + 0.5q$ & $d_2(q) = 7 - 0.8q$ \\
China (3) & $s_3(q) = 5 + q$ & $d_3(q)= 8 - q$ \\
\bottomrule
\end{tabular}
\end{table}

We analyze two trade policy scenarios:

\paragraph{Scenario 1 (Initial Trade Pattern).}
Initially, only China imposes a 10\% tariff on imports from the EU and the USA. 
The EU exports wine to the USA, and the USA exports wine to China.
There are $m=5$ active flows to determine, with a system of $m=5$ linear equations derived from \eqref{clearing}--\eqref{selling}. Consistency with \eqref{selection} is also verified. 
The computed quantities and prices (rounded to two decimal places) are summarized in Table \ref{tab:scenario1}.\footnote{%
The Mathematica code used for the computations in this example is available \href{https://github.com/paolopin/Supplementary-material-for-Network-Effects-of-Tariffs}{here}.
}

\begin{table}[h]
\centering
\caption{Scenario 1: Initial Trade Pattern}
\label{tab:scenario1}
\begin{tabular}{@{}lcc@{}}
\toprule
\textbf{Trade Flow} & \textbf{Quantity (million bottles)} & \textbf{Price (USD per bottle)} \\
\midrule
EU domestic sales ($q_{11}$) & 3.30 & 4.70 \\
EU exports to USA ($q_{21}$) & 2.11 & 4.70 \\
USA domestic sales ($q_{22}$) & 0.76 & 4.70 \\
USA exports to China ($q_{32}$) & 2.65 & $\frac{5.17}{1.1} = 4.70$ \\
China domestic sales ($q_{33}$) & 0.17 & 5.17 \\
\bottomrule
\end{tabular}

\vspace{0.4cm}

\begin{tabular}{@{}lc@{}}
\toprule
\textbf{USA Welfare Component} & \textbf{Value (million USD)} \\
\midrule
Consumer surplus & 2.90 \\
Firm profits & 3.29 \\
\bottomrule
\end{tabular}
\end{table}

The USA collects no government revenues in this scenario, as no tariffs are imposed on EU imports.

\paragraph{Scenario 2 (USA Introduces a Tariff).}
The USA now imposes a 20\% tariff on EU wine imports. In response, the EU redirects exports toward China, and US firms no longer find it profitable to export to China. 
Now, $m=4$ active flows determine the allocation. 
The computed quantities and prices (rounded to two decimal places) are summarized in Table \ref{tab:scenario2}.

\begin{table}[h]
\centering
\caption{Scenario 2: After USA Tariff}
\label{tab:scenario2}
\begin{tabular}{@{}lcc@{}}
\toprule
\textbf{Trade Flow} & \textbf{Quantity (million bottles)} & \textbf{Price (USD per bottle)} \\
\midrule
EU domestic sales ($q_{11}$) & 3.19 & 4.81 \\
EU exports to China ($q_{31}$) & 2.42 & $\frac{5.29}{1.1} \simeq 4.81$ \\
USA domestic sales ($q_{22}$) & 2.31 & 5.15 \\
China domestic sales ($q_{33}$) & 0.29 & 5.29 \\
\bottomrule
\end{tabular}

\vspace{0.4cm}

\begin{tabular}{@{}lc@{}}
\toprule
\textbf{USA Welfare Component} & \textbf{Value (million USD)} \\
\midrule
Consumer surplus & 2.13 \\
Firm profits & 1.33 \\
\bottomrule
\end{tabular}
\end{table}

Again, the USA collects no government revenues, as the EU stops exporting to the USA entirely.

\medskip

This illustrative example shows how tariffs can not only reduce domestic consumer welfare, but also harm domestic producers if trade patterns shift and competition in foreign markets intensifies. Moreover, it highlights how tariff policies may ultimately fail to generate additional revenue when the global trade network adjusts endogenously.





\section{Conclusion}
\label{conclusion}

This short paper develops a model in which country-specific tariffs shape trade flows, prices, and welfare in a global economy with a single homogeneous good. We characterize equilibrium trade flows under the assumptions of consumer market clearing and firm selection based on maximization of effective revenue. The resulting structure of trade is described by a Directed Acyclic Graph (DAG), with important implications for the formation of prices and the sensitivity of producers to external market conditions.

Tariffs affect welfare through three channels: consumer surplus, firm profits, and government revenues. When tariffs are adjusted marginally and the trade network remains fixed, welfare changes are relatively predictable: domestic firms are protected, consumers are harmed, exporters lose access to foreign markets, and government revenues respond ambiguously depending on price and quantity adjustments.

However, if tariff changes are large enough to reconfigure the network of active trade flows, welfare effects become highly non-linear and harder to predict. Discrete changes in the trade network can generate discontinuities in equilibrium prices, quantities, and welfare levels. Our fictitious numerical example, based on hypothetical trade flows between the EU, USA, and China, illustrates how the imposition of a tariff can inadvertently harm domestic consumers, fail to generate government revenues from tariffs, and expose domestic firms to intensified foreign competition.

An important direction for future research is to endogenize tariff-setting decisions by modeling them as a network formation game. In such a framework, countries would set tariffs strategically, not only to maximize their domestic welfare, but also to influence the structure of global trade links. This approach raises several possibilities: the existence of multiple network equilibria, the role of strategic patience in steering the system toward favorable outcomes, and the impact of international agreements not merely on efficiency, but also on equilibrium selection. Understanding the interplay between policy choices and network formation would significantly deepen our grasp of global trade dynamics.


\setcounter{section}{0} \global\long\def\thesection{\Alph{section}}
\setcounter{equation}{0} \global\long\def\theequation{\Alph{section}.\arabic{equation}}
\setcounter{proposition}{0} \global\long\def\theproposition{\Alph{proposition}}
\setcounter{lemma}{0} \global\long\def\thelemma{\Alph{lemma}}

\setlength{\abovedisplayskip}{9pt}   
\setlength{\belowdisplayskip}{9pt}    

\clearpage

\begin{center}
     \textbf{\Large Appendix}
     \vspace{-0.5cm}
\end{center}

\section{Proof of Proposition~\ref{prop:dag_existence} (page~\pageref{prop:dag_existence})}
\label{proof:prop_dag_existence}

\begin{proof}
\textbf{Existence.}
We define a fictitious one-shot interaction among representative firms, one for each country.\footnote{%
We are not formally defining a game, as firms behave as price takers and do not act strategically when choosing quantities, but only when choosing markets. One could think of it as a Cournot game with a continuum of firms within each country.}
There are \( n \) agents (representative firms).
For each agent \( i \), the action space is a column vector \( \bm{q}_i \in \mathbb{R}_+^n \), where \( q_{ji} \) denotes the quantity produced by firm \( i \) for market \( j \).
We impose the constraint 
\[
\bar{q}_{ji} = 
\begin{cases}
0 & \text{if } s_i(0) \geq d_j (0) \\
q \mbox{ such that } s_i(q) = d_j (q), & \mbox{otherwise}.
\end{cases}
\]
This is well defined, because functions are strictly monotonic for positive values, so that each $\bar{q}_{ji}$ is univocally defined.
Also, this threshold ensures that the overall action profile is a non-empty, compact, and convex subset of \( \mathbb{R}_+^{n \times n} \).

Given an action profile, each firm \( i \) computes the consumer prices \( p_{j,c}^T \) in each market \( j \), according to \eqref{clearing}, and its own marginal cost \( p_{i,f}^T \), according to \eqref{selling}.
Then, firm \( i \) \emph{best responds} as follows:
\[
q_{ji} =
\begin{cases}
0 & \text{if } \frac{p_{j,c}^T}{1+t_{ji}} < p_{i,f}^T, \\
\bar{q}_{ji} & \text{if } \frac{p_{j,c}^T}{1+t_{ji}} > p_{i,f}^T, \\
\text{any } q_{ji} \in \left[0, \bar{q}_{ji} \right] & \text{if } \frac{p_{j,c}^T}{1+t_{ji}} = p_{i,f}^T.
\end{cases}
\]

The resulting best--response correspondence is non-empty, convex-valued, and upper hemicontinuous.
Hence, by Kakutani's fixed-point theorem, a fixed point exists.
It is straightforward to verify that any fixed point satisfies conditions \eqref{clearing}--\eqref{selection}, and thus constitutes a tariff equilibrium.

\medskip
\textbf{DAG property.}
Suppose, for contradiction, that in equilibrium there exists a directed cycle of strictly positive trade flows among a set of countries \( \{i_1, \dots, i_m\} \), meaning that
\[
q_{i_1,i_2} > 0, \quad q_{i_2,i_3} > 0, \quad \dots, \quad q_{i_m,i_1} > 0.
\]
By the firm selection condition, for each \(k\), firms in country \(i_k\) choose to export to country \(i_{k+1}\) (modulo $m$) because it offers the maximal effective revenue, that is:
\[
\frac{p^T_{i_{k+1},c}}{1+t_{i_{k+1},i_k}} \geq \frac{p^T_{h,c}}{1+t_{h,i_k}} \quad \text{for all } h.
\]
In particular, comparing with \(i_k\)'s domestic market, we obtain:
\[
\frac{p^T_{i_{k+1},c}}{1+t_{i_{k+1},i_k}} \geq p^T_{i_{k},c}
\ \Rightarrow \
p^T_{i_{k+1},c} > p^T_{i_{k},c}.
\]
since all tariffs are positive.
However, this leads to a contradiction as it should hold for every $i_k \in \{i_1, \dots, i_m\}$ in the full cycle.

Therefore, no directed cycle of strictly positive flows can exist, and the network of positive trade flows forms a Directed Acyclic Graph (DAG).
This concludes the proof.
\end{proof}

\bibliographystyle{ecta}
\bibliography{biblio}

    \end{document}